\documentclass[
final
]{dmtcs-episciences}

\usepackage{xcolor}
\usepackage{graphicx}
\usepackage{tikz}

\usepackage[ruled,linesnumbered,vlined]{algorithm2e}

\SetCommentSty{mycommfont}
\SetKwFor{For}{for}{}{}
\SetKwFor{While}{while}{}{end while}%
\usepackage[title]{appendix}

\usepackage{xspace}
\tikzstyle{customnode}=[circle,inner sep=2, minimum size =3 pt, line width = 1pt, draw=black, fill=black]

\usepackage{todonotes}
\usepackage{amsthm,amsmath}
\newtheorem{theorem}{Theorem}
\newtheorem{lemma}[theorem]{Lemma}
\newtheorem{corollary}[theorem]{Corollary}

\theoremstyle{remark}

\newtheorem{claim}{Claim}[theorem]

\newcommand{\claimproof}{\noindent\emph{Proof of claim.} }
\usepackage{amssymb}

\def\cqedsymbol{\ifmmode$\lrcorner$\else{\unskip\nobreak\hfil
\penalty50\hskip1em\null\nobreak\hfil$\lrcorner$
\parfillskip=0pt\finalhyphendemerits=0\endgraf}\fi}

\title[Connected greedy colourings of perfect graphs]{Connected greedy colourings of perfect graphs and other classes: the good, the bad and the ugly}

\author[Laurent Beaudou et. al]{Laurent Beaudou\affiliationmark{1}
  \and Caroline Brosse\affiliationmark{1,2}\thanks{This author is supported by the French Agence Nationale de la Recherche under contract Digraphs ANR-19-CE48-0013-01.}~
  \and Oscar Defrain\affiliationmark{1,3}
  \and Florent Foucaud\affiliationmark{1}
    \and Aurélie Lagoutte\affiliationmark{1,4}
      \and Vincent Limouzy\affiliationmark{1}
        \and Lucas Pastor\affiliationmark{1}
  }

\affiliation{
  Universit\'e Clermont Auvergne, CNRS, Clermont Auvergne INP, Mines Saint-\'Etienne, LIMOS, Clermont-Ferrand, France\\
CNRS, Université C\^{o}te d’Azur, Inria, I3S, Sophia-Antipolis, France\\
Aix Marseille Université, Université de Toulon, CNRS, LIS, Marseille, France\\
Univ. Grenoble Alpes, CNRS, Grenoble INP, G-SCOP, Grenoble, France
}


\begin{document}
\publicationdata{vol. 25:2}{2023}{25}{10.46298/dmtcs.8715}{2021-11-17; 2021-11-17; 2022-11-18; 2023-11-14}{2023-11-19}

\maketitle

\begin{abstract}
  The Grundy number of a graph is the maximum number of colours
  used by the ``First-Fit'' greedy colouring algorithm over all vertex
  orderings.
  Given a vertex ordering $\sigma: v_1,\dots,v_n$, the ``First-Fit'' greedy colouring algorithm colours the vertices in the order of $\sigma$ by assigning to each vertex the smallest colour unused in its neighbourhood.
  By restricting this procedure to vertex orderings that are connected, we obtain {\em connected greedy colourings}.
  For some graphs, all connected greedy colourings use exactly
  $\chi(G)$ colours; they are called {\em good graphs}. On the
  opposite, some graphs do not admit any connected greedy
  colouring using only $\chi(G)$ colours; they are called {\em
    ugly graphs}. 
    We show that no perfect graph is ugly.
    We also give simple proofs of this fact for subclasses of perfect graphs (block graphs, comparability graphs), and show that no $K_4$-minor-free graph is ugly. Moreover, our proofs are constructive, and imply the existence of polynomial-time algorithms to compute good connected orderings for these graph classes.
\end{abstract}

\section{Introduction}

Optimally colouring a graph has been and remains a hard task:
Karp~\cite{karp} lists the \textsc{Chromatic Number} problem among
his twenty-one \textsc{NP}-hard problems in 1972. Facing hard
problems, a common tactic consists in solving them for subclasses
of graphs, but even deciding if a planar graph of maximum degree 4
admits a 3-colouring is an \textsc{NP}-complete
problem~\cite[Section 2]{garey_johnson_stockmeyer}. To deal with
graph colourings and their applications, heuristics have 
been
designed. Greedy colouring, also called ``First-Fit'', is among the first heuristics that come
to mind.

\paragraph{Greedy colouring.}
A {\em greedy colouring} of a graph $G$ relative to an ordering
$\sigma: v_1,v_2,\ldots,v_n$ of its vertices is obtained by
colouring the vertices in the order of $\sigma$ and assigning to
each vertex the smallest positive integer that is unused in its neighbourhood.
Let $\chi(G)$ denote the chromatic number of the graph $G$ and let
$\chi(G,\sigma)$ denote the number of colours used when
colouring $G$ greedily with respect to the ordering $\sigma$. Since any greedy
colouring is proper (no two adjacent vertices have the same
colour), we may observe that $\chi(G) \leq \chi(G,\sigma)$ for any
ordering $\sigma$ of vertices of $G$. Actually the chromatic
number is always attained by some ordering (we call such orderings \emph{good}). By noting
$\mathcal{S}(G)$ the set of orderings on the vertices of $G$, we
have
\begin{equation}
  \chi(G) = \min \{ \chi(G,\sigma): \sigma \in \mathcal{S}(G)\}.
\end{equation}
To see this, it is enough to consider an optimal colouring of $G$,
thus using colours $\{1,\ldots,\chi(G)\}$ and take any ordering
$\sigma$ which ranks vertices with respect to their colours (first
all the vertices coloured with 1, then with 2 and so on). Following
this order, no vertex receives a colour strictly larger than the
one assigned by the optimal colouring.

\paragraph{Grundy number.}
Although greedy colourings have a chance to perform well, choosing
$\sigma$ with no care could lead to bad choices. The {\em Grundy
  number} of a graph $G$, denoted by $\Gamma(G)$, is a measure of
the worst possible choice among greedy colourings. It is the
largest number of colours used among all greedy colourings:
\begin{equation}
  \Gamma(G) := \max \{ \chi(G,\sigma): \sigma \in \mathcal{S}(G)\}.
\end{equation}
Greedy colourings have been called Grundy colourings by several
authors referring to a note on combinatorial games by
Grundy~\cite{grundy} from 1939. Forty years later, Christen and
Selkow~\cite{cographs} introduced the Grundy number. They proved
that for a graph $G$, we have $\Gamma(H) = \chi(H)$ for all
induced subgraphs $H$ of $G$ if and only if $G$ is a cograph. Note
that the Grundy number of a graph may be arbitrarily larger than
its chromatic number (for any fixed $n$, removing a perfect
matching from the complete bipartite graph $K_{n,n}$ yields a
graph $G_n$ for which $\chi(G_n) = 2$ and $\Gamma(G_n) = n$).

\paragraph{Connected orderings.}
An ordering $\sigma: v_1,v_2,\ldots,v_n$ of the vertices of a (connected)
graph $G$ is called a {\em connected ordering} if for each integer
$i$ between 1 and $n$, the subgraph induced by the vertices
$v_1,\ldots,v_i$ is connected. Greedy colourings using these
connected orderings have been studied about thirty years ago by Hertz and De
Werra~\cite{HdW89} and by Babel and Tinhofer~\cite{BT94}.  Let
$\mathcal{S}_c(G)$ be the set of connected orderings of a graph
$G$ and define the {\em connected greedy chromatic number} of a
 connected
graph $G$, denoted $\chi_c(G)$, as the minimum number of colours
used for connected orderings:
\begin{align*}
  \chi_c(G) := \min \{ \chi(G,\sigma): \sigma \in \mathcal{S}_c(G)\}.
\end{align*}
In general, $\chi_c(G)$ is not equal to $\chi(G)$; see~\cite[Theorem 2]{latin14}. We
similarly define the {\em connected Grundy number} of a connected graph
$G$, denoted $\Gamma_c(G)$, as the maximum number of colours for
connected orderings:
\begin{align*}
  \Gamma_c(G) := \max \{ \chi(G,\sigma): \sigma \in \mathcal{S}_c(G)\}.
\end{align*}
Note that for any connected graph $G$, we have the following chain of
inequalities:
\begin{align*}
  \chi(G) \leq \chi_c(G) \leq \Gamma_c(G) \leq \Gamma(G).
\end{align*}
Benevides, Campos, Dourado, Griffiths, Morris, Sampaio and
Silva~\cite{latin14} have recently proven that $\chi_c(G)$ cannot
be arbitrarily large with respect to $\chi(G)$. The difference can
be at most 1: $\chi_c(G) \leq \chi(G)+1$, see~\cite[Theorem 3]{latin14}.

\paragraph{Introducing the good, the bad and the ugly.}
Following the terminology of Le and Trotignon~\cite{clawfree}, we call a connected
graph $G$ satisfying $\chi(G)=\Gamma_c(G)$ \emph{good}, that is, $G$ is a
graph for which any connected ordering is good (gives an optimal
colouring).\footnote{Actually, in~\cite{clawfree} a graph $G$ is
  called good only if $\chi(H)=\Gamma_c(H)$ for every connected
  induced subgraph $H$ of $G$. In this paper, we consider only
  hereditary classes of graphs and we are interested in
  determining whether all graphs in the class are good or not, so
  this difference in the definition is irrelevant in our
  context.} All other connected graphs are called \emph{bad}.\footnote{Bad
  graphs were called \emph{slightly hard-to-colour}
  in~\cite{BT94}.} 
A connected graph $G$ for which \emph{no} connected ordering achieves
the optimal value $\chi(G)$, i.e.~$\chi_c(G)>\chi(G)$, is called \emph{ugly}.\footnote{Ugly
  graphs were called \emph{globally hard-to-colour} in~\cite{BT94},
  but we prefer to follow the lines of the less lengthy terms
  of~\cite{clawfree}.} For any connected graph, a connected ordering of its vertices that yields an optimal colouring is called a \emph{good connected ordering}.

\paragraph{Known results.} It can be
observed that all bipartite graphs are
good~\cite{latin14}. In~\cite{HdW89}, Hertz and De Werra showed that all
fish-free parity graphs are good. The fish and the gem graphs are bad; see Figure~\ref{fig:fish}. In~\cite{clawfree}, the authors
characterized good claw-free graphs in terms of forbidden induced
subgraphs.

\begin{figure}[ht!]
  \centering
  \scalebox{0.6}{\begin{tikzpicture}
  \begin{scope}
        \node[customnode](1) at (0,0) {};
        \draw (1)+(0,-0.5) node {$v_1$};
        \draw (1)+(0,2) node[customnode](2) {};
        \draw (2)+(0,0.5) node {$v_2$};
        \draw (1)+(1,1) node[customnode](3) {};
        \draw (3)+(0,0.5) node {$v_3$};
        \draw (3)+(1,1) node[customnode](4) {};
        \draw (4)+(0,0.5) node {$v_4$};
        \draw (4)+(1,-1) node[customnode](5) {};
        \draw (5)+(0.5,0) node {$v_5$};
        \draw (4)+(0,-2) node[customnode](6) {};
        \draw (6)+(0,-0.5) node {$v_6$};
        \draw (1)--(2)--(3)--(1);
        \draw (3)--(4)--(5)--(6);
        \draw (3)--(6)--(4);
        
    \end{scope}

    \begin{scope}[xshift=8cm]
        
        \node[customnode](1) at (0,0) {};
        \draw (1)+(0,-0.5) node {$v_1$};
        \draw (1)+(-2,0) node[customnode](2) {};
        \draw (2)+(0,-0.5) node {$v_2$};
        \draw (1)+(2,0) node[customnode](3) {};
        \draw (3)+(0,-0.5) node {$v_3$};
        \draw (1)+(-1,2) node[customnode](4) {};
        \draw (4)+(0,0.5) node {$v_4$};
        \draw (1)+(1,2) node[customnode](5) {};
        \draw (5)+(0,0.5) node {$v_5$};
        \draw (1)--(2)--(4)--(5)--(3)--(1);
        \draw (4)--(1)--(5);
    \end{scope}
        
  \end{tikzpicture}}
  \caption{The fish and the gem, two bad graphs discovered in \cite{BT94} and~\cite{HdW89} (bad connected vertex-orderings are $v_1,\ldots,v_n$).}\label{fig:fish}  
\end{figure}
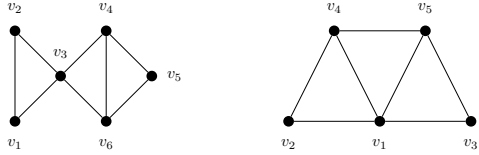

A planar cubic ugly graph was presented in~\cite{BT94}; see
Figure~\ref{fig:ugly}. A claw-free ugly graph was also found
in~\cite{clawfree} (in fact it is a line graph of a multigraph), and it can be modified to obtain an ugly line graph, see Figure~\ref{fig:ugly-clawfree}. These examples have triangles, but one can obtain ugly graphs of arbitrarily large girth. Indeed, the building blocks of these examples are gadgets (here, diamonds) in which two specified vertices must receive the same colour in any optimal colouring, and such gadgets of arbitrarily large girth can be obtained by taking colour-edge-critical graphs of large girth and deleting an edge (the two endpoints of that edge now need to receive the same colour in any optimal colouring).

\begin{figure}[ht!]
  \centering
  \scalebox{0.5}{\begin{tikzpicture}

      \node[customnode](0) at (0.5,0) {};
      \node[customnode](1) at (1,1) {};
      \node[customnode](2) at (1,-1) {};
      \node[customnode](3) at (2,1.5) {};
      \node[customnode](4) at (2,0.5) {};
      \node[customnode](5) at (2,-0.5) {};
      \node[customnode](6) at (2,-1.5) {};
      \node[customnode](7) at (3,1) {};
      \node[customnode](8) at (3,-1) {};

      \draw (0)--(1)--(3)--(4)--(1);
      \draw (0)--(2)--(5)--(8);
      \draw (3)--(7)--(4); 
      \draw (7)--(8)--(6); 
      \draw (2)--(6)--(5);

      \node[customnode](9) at (-0.5,0) {};
      \node[customnode](10) at (-1,1) {};
      \node[customnode](11) at (-1,-1) {};
      \node[customnode](12) at (-2,1.5) {};
      \node[customnode](13) at (-2,0.5) {};
      \node[customnode](14) at (-2,-0.5) {};
      \node[customnode](15) at (-2,-1.5) {};
      \node[customnode](16) at (-3,1) {};
      \node[customnode](17) at (-3,-1) {};   

      \draw (0)--(9)--(10)--(12)--(13)--(10);
      \draw (9)--(11)--(14)--(17);
      \draw (12)--(16)--(13); 
      \draw (16)--(17)--(15); 
      \draw (11)--(15)--(14);
     
  \end{tikzpicture}}
  \caption{An ugly planar cubic graph from~\cite{BT94}.}\label{fig:ugly}  
\end{figure}
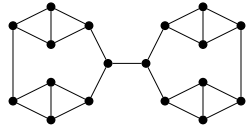

\begin{figure}[ht!]
  \centering
\scalebox{0.4}{\begin{tikzpicture}

\begin{scope}

    \node[customnode](1a) at (0.5,0) {};
    \draw (1a)+(-1,0) node[customnode](1c) {};
    \draw (1c)--(1a);
    
    \draw (1a)+(1,1) node[customnode](2a) {};
    \draw (2a)+(0.25,1.5) node[customnode](3a) {};
    \draw (2a)+(1.5,0.25) node[customnode](4a) {};
    \draw (3a)+(1,1) node[customnode](5a) {};
    \draw (4a)+(1,1) node[customnode](6a) {};
    \draw (2a)+(1.5,1.5) node[customnode](7a) {};
    \draw (2a)+(2.25,2.25) node[customnode](8a) {};
    \draw (1a)--(2a)--(3a)--(5a);
    \draw (2a)--(4a)--(6a);
    \draw (3a)--(4a);
    \draw (5a)--(7a)--(6a)--(8a)--(5a) (7a)--(8a);

    \draw (1a)+(1,-1) node[customnode](2b) {};
    \draw (2b)+(0.25,-1.5) node[customnode](3b) {};
    \draw (2b)+(1.5,-0.25) node[customnode](4b) {};
    \draw (3b)+(1,-1) node[customnode](5b) {};
    \draw (4b)+(1,-1) node[customnode](6b) {};
    \draw (2b)+(1.5,-1.5) node[customnode](7b) {};
    \draw (2b)+(2.25,-2.25) node[customnode](8b) {};
    \draw (1a)--(2b)--(3b)--(5b);
    \draw (2b)--(4b)--(6b);
    \draw (3b)--(4b);
    \draw (5b)--(7b)--(6b)--(8b)--(5b) (7b)--(8b);

    \draw (2b)--(2a);

    \draw (1c)+(-1,1) node[customnode](2c) {};
    \draw (2c)+(-0.25,1.5) node[customnode](3c) {};
    \draw (2c)+(-1.5,0.25) node[customnode](4c) {};
    \draw (3c)+(-1,1) node[customnode](5c) {};
    \draw (4c)+(-1,1) node[customnode](6c) {};
    \draw (2c)+(-1.5,1.5) node[customnode](7c) {};
    \draw (2c)+(-2.25,2.25) node[customnode](8c) {};
    \draw (1c)--(2c)--(3c)--(5c);
    \draw (2c)--(4c)--(6c);
    \draw (3c)--(4c);
    \draw (5c)--(7c)--(6c)--(8c)--(5c) (7c)--(8c);

    \draw (1c)+(-1,-1) node[customnode](2d) {};
    \draw (2d)+(-0.25,-1.5) node[customnode](3d) {};
    \draw (2d)+(-1.5,-0.25) node[customnode](4d) {};
    \draw (2d)+(-1.5,-1.5) node[customnode](7d) {};
    \draw (2d)+(-2.25,-2.25) node[customnode](8d) {};
    \draw (3d)+(-1,-1) node[customnode](5d) {};
    \draw (4d)+(-1,-1) node[customnode](6d) {};
    \draw (1c)--(2d)--(3d)--(5d);
    \draw (2d)--(4d)--(6d);
    \draw (3d)--(4d);
    \draw (5d)--(7d)--(6d)--(8d)--(5d) (7d)--(8d);

    \draw (2c)--(2d);
    
    \draw (0,-5) node {{\LARGE(a)}};

\end{scope}

\begin{scope}[xshift=11cm]
    \node[customnode](1a) at (0.5,0) {};
    \draw (1a)+(-1,0) node[customnode](1c) {};
    \draw (1c)--(1a);
    
    \draw (1a)+(1,1) node[customnode](2a) {};
    \draw (2a)+(0.25,1.5) node[customnode](3a) {};
    \draw (2a)+(1.5,0.25) node[customnode](4a) {};
    \draw (3a)+(1,1.25) node[customnode](5a) {};
    \draw (4a)+(1.25,1) node[customnode](6a) {};
    \draw (2a)+(1.25,1.75) node[customnode](7a) {};
    \draw (2a)+(2.25,2.75) node[customnode](8a) {};
    \draw (2a)+(1.75,1.25) node[customnode](9a) {};
    \draw (2a)+(2.75,2.25) node[customnode](10a) {};
    \draw (2a)+(2,2) node[customnode](11a) {};
    \draw (1a)--(2a)--(3a)--(5a);
    \draw (2a)--(4a)--(6a);
    \draw (3a)--(4a);
    \draw (5a)--(7a)--(8a)--(5a) (6a)--(9a)--(10a)--(6a) (9a)--(7a)--(11a)--(8a) (9a)--(11a)--(10a)--(8a);

    \draw (1a)+(1,-1) node[customnode](2b) {};
    \draw (2b)+(0.25,-1.5) node[customnode](3b) {};
    \draw (2b)+(1.5,-0.25) node[customnode](4b) {};
    \draw (3b)+(1,-1.25) node[customnode](5b) {};
    \draw (4b)+(1.25,-1) node[customnode](6b) {};
    \draw (2b)+(1.25,-1.75) node[customnode](7b) {};
    \draw (2b)+(2.25,-2.75) node[customnode](8b) {};
    \draw (2b)+(1.75,-1.25) node[customnode](9b) {};
    \draw (2b)+(2.75,-2.25) node[customnode](10b) {};
    \draw (2b)+(2,-2) node[customnode](11b) {};
    \draw (1a)--(2b)--(3b)--(5b);
    \draw (2b)--(4b)--(6b);
    \draw (3b)--(4b);
    \draw (5b)--(7b)--(8b)--(5b) (6b)--(9b)--(10b)--(6b) (9b)--(7b)--(11b)--(8b) (9b)--(11b)--(10b)--(8b);

    \draw (2b)--(2a);

    \draw (1c)+(-1,1) node[customnode](2c) {};
    \draw (2c)+(-0.25,1.5) node[customnode](3c) {};
    \draw (2c)+(-1.5,0.25) node[customnode](4c) {};
    \draw (3c)+(-1,1.25) node[customnode](5c) {};
    \draw (4c)+(-1.25,1) node[customnode](6c) {};
    \draw (2c)+(-1.25,1.75) node[customnode](7c) {};
    \draw (2c)+(-2.25,2.75) node[customnode](8c) {};
    \draw (2c)+(-1.75,1.25) node[customnode](9c) {};
    \draw (2c)+(-2.75,2.25) node[customnode](10c) {};
    \draw (2c)+(-2,2) node[customnode](11c) {};
    \draw (1c)--(2c)--(3c)--(5c);
    \draw (2c)--(4c)--(6c);
    \draw (3c)--(4c);
    \draw  (5c)--(7c)--(8c)--(5c) (6c)--(9c)--(10c)--(6c) (9c)--(7c)--(11c)--(8c) (9c)--(11c)--(10c)--(8c);

    \draw (1c)+(-1,-1) node[customnode](2d) {};
    \draw (2d)+(-0.25,-1.5) node[customnode](3d) {};
    \draw (2d)+(-1.5,-0.25) node[customnode](4d) {};
    \draw (2d)+(-1.25,-1.75) node[customnode](7d) {};
    \draw (2d)+(-2.25,-2.75) node[customnode](8d) {};
    \draw (3d)+(-1,-1.25) node[customnode](5d) {};
    \draw (4d)+(-1.25,-1) node[customnode](6d) {};
    \draw (2d)+(-1.75,-1.25) node[customnode](9d) {};
    \draw (2d)+(-2.75,-2.25) node[customnode](10d) {};
    \draw (2d)+(-2,-2) node[customnode](11d) {};
    \draw (1c)--(2d)--(3d)--(5d);
    \draw (2d)--(4d)--(6d);
    \draw (3d)--(4d);
    \draw (5d)--(7d)--(8d)--(5d) (6d)--(9d)--(10d)--(6d) (9d)--(7d)--(11d)--(8d) (9d)--(11d)--(10d)--(8d);

    \draw (2c)--(2d);
    
    \draw (0,-5) node {{\LARGE(b)}};

\end{scope}

\end{tikzpicture}}
\caption{(a) An ugly planar claw-free graph from~\cite{clawfree}.\qquad (b) An ugly planar line graph.}\label{fig:ugly-clawfree}  
\end{figure}
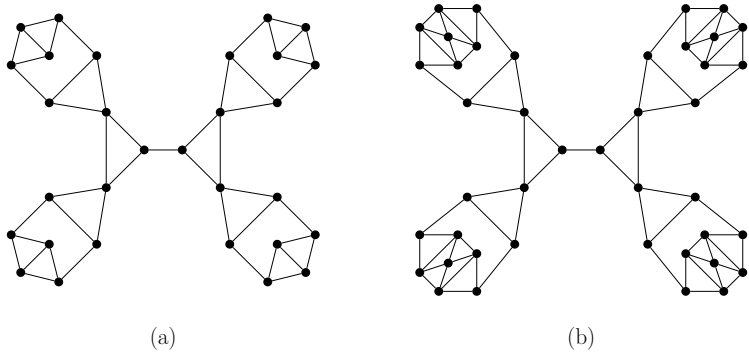

Clearly, every ugly graph is bad. It is
co\textsc{NP}-hard to recognize ugly graphs~\cite{latin14}, even
for inputs that are line graphs, or $H$-free with $H$ not a linear forest, or $H$
containing an induced $P_9$~\cite{H-free}. (This implies that for
any such $H$, there exist $H$-free ugly graphs.) On the other hand, it is proved in~\cite{H-free} that for any $H$ that is an induced subgraph of $P_4+K_1$ or $P_5$, there are no $H$-free ugly graphs.

A graph is \emph{perfect} if for each of its induced subgraphs, the chromatic number equals the clique number (size of the largest clique). Colouring perfect graphs has been studied for decades. For some subclasses of perfect graphs, the classic colouring algorithms actually work greedily on a connected ordering. For example, an ordering $\sigma$ of the vertices of a graph $G$ is \emph{perfect} if for every induced subgraph $H$ of $G$, the
sub-ordering $\sigma_H$ of $\sigma$ induced by $V(H)$ gives
$\chi(H,\sigma)=\chi(H)$~\cite{perfectly-orderable}. Graphs with such orderings are called \emph{perfectly orderable}; they are perfect and include all chordal graphs and all comparability
graphs. An ordering of $G$ is
called a \emph{perfect elimination ordering} if for every $i$ with
$i<n$, the neighbours of $v_i$ among $\{v_{i+1},\ldots,v_n\}$ form
a clique in $G$. A graph is known to be chordal if and only if it
admits a perfect elimination ordering, and such an ordering may be
found in linear time $O(m+n)$ (where $n$ is the number of vertices and $m$ is the number of edges) as the reversed order of the LexBFS algorithm~\cite{chordal}. It follows that for a graph $G$, there exists a perfect
elimination ordering $\sigma$ of $V(G)$ whose reverse $\sigma'$ is
connected (since it corresponds to a BFS order) and gives
$\chi_c(G,\sigma')=\chi(G)$, and this is also a perfect ordering. Hence no chordal graph is ugly. An extension of this concept is the one of a \emph{semi-perfect elimination
  ordering} (see~\cite{semiperfect} for a definition). It is
proved in~\cite{semiperfect} that every vertex-ordering of $V(G)$
that is the reverse of a LexBFS ordering (and thus, connected) is a semi-perfect
elimination ordering if and only if $G$ has no house, no hole and
no domino as an induced subgraph. Such graphs are called
\emph{HHD-free graphs}; all chordal graphs and
distance-hereditary graphs are HHD-free. It was proved
in~\cite{semiperfect-are-perfect} that the reverse of a
semi-perfect elimination ordering is a perfect ordering, and thus, no HHD-free graph is ugly.
Yet another larger class of perfect graphs (containing HHD-free graphs) with connected orderings is the class of \emph{Meyniel graphs}
(graphs where every odd cycle of length at least~$5$ has at least
two chords). 
In~\cite{meyniel}, a LexBFS-like $O(n^2)$ algorithm that 
produces a connected ordering $\sigma$ of
the vertices such that $\chi(G,\sigma)=\chi(G)$ is given.

Recently, connected greedy edge-colourings (equivalently, connected greedy colourings of line graphs) have been studied in~\cite{bonamy_21}, and it was proved that there is no line graph of a bipartite graph that is ugly.\footnote{Moreover, a careful analysis of the proof of~\cite{bonamy_21} gives an algorithm running in time $O(n^4)$ to compute a good connected ordering of any connected line graph of bipartite graph on $n$ vertices.} Such graphs are perfect.

\paragraph{Our results.} 
In this paper, we continue the hunt for graph classes containing only good connected graphs, and for graph classes containing no ugly graphs. For a graph class of the latter type, given a graph $G$ of this class, we note that deciding whether $\chi_c(G)=\chi(G)$ is trivially polynomial-time solvable (always say ''yes''). Thus our work is related to the algorithmic work from~\cite{H-free}.

We first show how to inductively create new good graphs out of good graphs in Section~\ref{sec:great}, with an application to cactus graphs and block graphs (a \emph{block graph} is a graph in which every
biconnected component induces a clique, and a \emph{cactus graph} is a
graph in which every biconnected component induces a cycle). Using the inductive structure of $K_4$-minor-free graphs in a similar manner, we then show (constructively) in Section~\ref{sec:SP} that no member of this class is ugly. We also show constructively that no comparability graph is ugly in Section~\ref{sec:comp}. Finally, our main theorem is to generalize several known results about subclasses of perfect graphs by showing that no perfect graph is ugly (and a good connected ordering of a perfect graph can be computed in polynomial time, using an algorithm for perfect graph colouring as a black box). This is done in Section~\ref{sec:perfect}.

\paragraph{Definitions and notations.} For standard definitions and notations of graph theory that 
are not recalled in this article, we refer the reader to \cite{diestel}. 
A \emph{(proper) $k$-colouring} of a graph $G=(V,E)$ is a function $c:V\to \{1, \ldots, k\}$ such that $c(u)\neq c(v)$ whenever $uv\in E$. A graph is \emph{$k$-colourable} if it admits a proper $k$-colouring. Its \emph{chromatic number} $\chi(G)$ is the smallest integer $k$ such that $G$ is $k$-colourable, and we call \emph{optimal} colouring any $\chi(G)$-colouring. A graph $G$ is \emph{$k$-chromatic} if $k=\chi(G)$.
A graph is \emph{bipartite} if it is 2-colourable. 

Given a graph $G=(V,E)$ and a vertex $v$, the \emph{neighbourhood} of $v$ is the set $N(v)=\{u\in V \ | \ uv\in E\}$, and we call \emph{neighbours} of $v$ the elements in such a set. For a set of vertices $S$, we denote $N(S)=\left(\cup_{v\in S} N(v)\right) \setminus S$.
A \emph{clique} in $G$ is a set of pairwise adjacent vertices and a \emph{$k$-clique} is a clique of size $k$. An \emph{independent set}
is a set of pairwise non-adjacent vertices.
Given a subset of vertices $X$, the \emph{subgraph induced by} $X$, denoted $G[X]$, is the graph $(X, E\cap(X\times X))$ obtained from $G$ by removing the vertices that are not in $X$. On the other hand $G-X$ is the subgraph $G[V\setminus X]$ induced by $V\setminus X$. When $X=\{v\}$ for some $v\in V$, we may write $G-v$ instead of $G-\{v\}$. A \emph{subgraph} of $G$ is a graph $G'=(V',E')$ such that $V'\subseteq V$ and $E'\subseteq E$.
An \emph{orientation} of $G$ is a directed graph obtained from $G$ by transforming each edge $uv\in E$ into either the arc $u\to v$ or the arc $v\to u$. An orientation is \emph{acyclic} if it contains no directed cycle.

\section{Making good graphs out of good graphs}\label{sec:great}

In this section, we show a natural way of building new good
graphs by gluing them through a cut-vertex. 
A vertex $v$ in a connected graph $G$ is a \emph{cut-vertex} if its removal disconnects $G$. A \emph{biconnected} graph is a connected graph  without any cut-vertex. A \emph{biconnected component} of a graph $G$ is an inclusion-wise maximal set of vertices inducing a biconnected graph.

In order to get the desired result, we will need to
strengthen the hypothesis and introduce for that purpose great graphs. A {\em great graph} is a
connected graph $G$ such that for every connected ordering
$\sigma:v_1, \ldots, v_n$ of its vertices and every positive integer
$i$, we may colour vertex $v_1$ with colour $i$, apply the greedy
colouring algorithm to $v_2,\ldots,v_n$ and only use colours
between 1 and $\chi(G)$ among the vertices $v_2,\ldots,v_n$. Of
course, a great graph is also good since $v_1$ may be coloured
$1$. Notice that complete graphs, bipartite graphs and cycles are
great. For great graphs, we can thus find a good connected ordering in linear time using a standard graph traversal algorithm such as Depth-First Search.

\begin{lemma}\label{lemm:extragood-biconnectedcomponents}
If all biconnected components of a connected graph $G$ induce a great graph,
then $G$ is great.
\end{lemma}
\begin{proof}
We proceed by induction on the number of biconnected components. Let $v$ be
a cut-vertex of~$G$. Let $G_1,\ldots,G_k$ be the subgraphs induced
by the connected components of $G-v$, together with $v$, meaning that $G_i=G[V_i\cup\{v\}]$ where $V_i$ is a connected component of $G-v$. 
Observe that each $G_i$ has strictly less biconnected components than $G$.
Let
$\sigma:v_1,\ldots, v_n$ be a connected ordering of $G$ and assume
without loss of generality that $v_1$ belongs to $G_1$. We
consider the sub-orderings $\sigma_1,\ldots,\sigma_k$ of $\sigma$,
where $\sigma_i$ contains only the vertices of $G_i$. Note that each ordering
$\sigma_i$ is a connected ordering of $G_i$, starting with $v_1$ if $i=1$, or with $v$ if $i>1$. 
Let us call $c$ 
the greedy colouring relative to
$\sigma$ starting from any colour $\alpha$ on $v_1$.
Then $c$ will agree on each $G_i$ with the greedy colouring relative to $\sigma_i$ starting from $c(v)$ on $v$ if $i>1$, or starting from $\alpha$ on $v_1$ if $i=1$.
Since each
$G_i$ is great by induction hypothesis, 
this colouring
will not use more
than $\max\{\chi(G_i): 1\leq i \leq k\}$ colours, which is equal
to $\chi(G)$.
\end{proof}

\begin{corollary}\label{cor:block-cactus-lineperfect}
    Every connected block graph and every connected cactus graph is great.
\end{corollary}

\section{Classes of graphs with no ugly member}

In this section, we exhibit two
classes of non-ugly graphs, i.e.~classes of graphs admitting good connected orderings: the class of $K_4$-minor-free graphs and the
class of perfect graphs. We also give a simple and constructive proof for comparability graphs (which are perfect). Note that there exist bad graphs in these graph classes, consider for example the fish graph, which is $K_4$-minor-free and comparability; see Figure~\ref{fig:fish}.

\subsection{$K_4$-minor-free graphs}\label{sec:SP}

A graph $H$ is a \emph{minor} of $G$ if $H$ can be obtained from 
$G$ by a series of vertex deletions, edge deletions, 
edge contractions (replacing two adjacent vertices $u$, $v$ by a 
single vertex adjacent to all neighbours of $u$ and $v$). A graph 
$G$ is \emph{$K_4$-minor-free} if $K_4$ is not a minor of $G$.

The class of $K_4$-minor-free graphs has been extensively studied in
different contexts and inherited many different names such as Series-Parallel
graphs, partial 2-trees, or graphs with treewidth at most 2 \cite{bod}. We shall observe a simple
fact about those graphs that will help us in the search of a good
ordering. For this we need to define the notion of $2$-tree. A {\em
  $2$-tree} is any graph obtained from $K_3$ and then
repeatedly adding vertices in such a way that each added vertex 
has exactly two neighbours which are adjacent to each other. A graph
$G$ is $K_4$-minor-free if and only if it is a \emph{partial 2-tree}, that is, a subgraph of a
2-tree~\cite{bod}. Such graphs are easily seen to be 3-colourable~\cite{duffin}.

\begin{lemma}
  Any $K_4$-minor-free graph $G$ on at least two vertices has two vertices of degree at most~2.
  \label{lem:deg2}
\end{lemma}

\begin{proof}
  Let $G$ be a $K_4$-minor-free graph. Let $\hat{G}$ be a 2-tree
  having $G$ as a subgraph. Since $\hat{G}$ is chordal and has more than two vertices, it is known to have
  two simplicial vertices~\cite{D61}. Since the maximum clique has order at most 3 in $\hat{G}$, these
  two simplicial vertices have degree at most 2 in
  $\hat{G}$. Therefore, they have degree at most 2 in $G$.
\end{proof}

\begin{lemma}
Any connected $K_4$-minor-free graph $G$ on at least two vertices has a vertex of degree at most 2 whose removal leaves $G$
connected. \label{lem:2conn}
\end{lemma}
\begin{proof}
  If there is a vertex of degree 1, then its removal definitely leaves
  $G$ connected. So, we may assume that the minimum degree of $G$ is~2.
  
  For a contradiction, suppose that every vertex of degree 2
  disconnects the graph. For each such vertex, let us look at the
  number of vertices in
  the smallest connected component of $G$ after its
  removal: let $u$ be a vertex of degree~2 minimizing this
  quantity. Vertex $u$ has two neighbours $x$ and $y$. Without loss
  of generality, we may assume that in $G-u$, the component containing
  $x$ is the smallest. Let us call the graph induced by this component
  $G_x$. Graph $G_x$ has at least two vertices (or $x$ has degree~1 in
  $G$). By Lemma~\ref{lem:deg2}, $G_x$ has two vertices of degree at
  most~2, one of which is distinct from $x$. Let it be $z$. Observe
  that $z$ has same degree in $G_x$ and in $G$. So $z$ has degree~2 in $G$
  and by our assumption, it must disconnect $G$. One of the connected components after
  removal of $z$ from $G$ has fewer vertices than $G_x$, which is a contradiction to the choice of $u$.
\end{proof}

\begin{theorem}\label{thm:SPnotugly}
  No $K_4$-minor-free graph is ugly, and a good connected ordering of any connected $K_4$-minor-free graph on $n$ vertices can be computed in time $O(n^2)$.
\end{theorem}

\begin{proof}
  We prove the first part of the statement by induction on the number of vertices of $G$.

  If $G$ has two vertices or less, then it is trivial.
  Let us consider a connected $K_4$-minor-free graph $G$ on $n$ vertices, $n \geq 3$. 
  If $G$ is bipartite, then it is good and hence any
  connected ordering gives a 2-colouring. Let us thus assume that $G$ is
  not bipartite. Since $G$ is 3-colourable~\cite{duffin}, it implies that $G$ has
  chromatic number~3. By Lemma~\ref{lem:2conn}, there is a vertex $u$
  of degree at most~2 whose removal gives a connected graph on $n-1$
  vertices. By the induction hypothesis, there is a connected ordering of
  $G-u$ which yields an optimal colouring of $G-u$. By adding $u$ at the
  end of this ordering, we obtain an optimal colouring of $G$ (since $u$ has
  at most two neighbours, it receives a colour among 1, 2 and 3).
  
  To obtain a good connected ordering in polynomial time, if $G$ is bipartite, one can use any connected ordering (this can be done in linear time using a standard graph traversal algorithm). Otherwise, it suffices to iteratively find a vertex of degree at most~2 whose removal gives a connected graph, and reverse the obtained ordering. This process can be done in $O(n^2)$ time.
\end{proof}

\subsection{Comparability graphs}\label{sec:comp}

This section is devoted to proving that the class of comparability graphs is not ugly.
A graph $G=(V,E)$  is a \emph{comparability graph} if there exists an acyclic orientation of $G$ that is transitive.
An orientation is called \emph{transitive} if it contains the arc $a\rightarrow c$ whenever it contains the arcs $a\rightarrow b$ and $b\rightarrow c$. 
The class of comparability graphs forms a strict subclass of perfect graphs \cite{BOOKgolumbic}.

A comparability graph, together with a transitive orientation, naturally encodes the relations
of a partially ordered set (\emph{poset} for short). A \emph{poset} $P=(V,\prec)$ is a binary relation $\prec$ defined
on a ground set $V$ that is reflexive ($\forall x\in V,\ x\prec x$), 
anti-symmetric ($\forall x,y \in V\times V,\ {(x\prec  y)}\text{ and }(y\prec x) \Rightarrow x=y)$ 
and transitive ($x\prec y\text{ and }y \prec z \Rightarrow x \prec  z$).
It is called a \emph{total order} if $\forall x,y \in V\times V$ we have either $x\prec y$ or $y\prec x$.
A comparability graph and one of its transitive orientations are presented in Figure~\ref{fig:Poset}.
An element $x\in V$ is called \emph{maximal} (resp.~\emph{minimal}) if there exists no element $y\in V\setminus\{x\}$ such that $x \prec y$ (resp.~$y\prec x$). A \emph{chain} in a poset is 
a set of elements that induces a total order. Let us remark that a chain corresponds to a complete graph in the associated comparability graph.

\begin{figure}[h!]
    \begin{center}
    
    \includegraphics[page=1,scale=1]{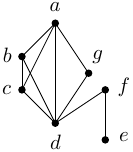}
    \hspace{3cm}
    \includegraphics[page=2,scale=1]{posets}
    \end{center}
    \caption{A comparability graph (left), and a transitive orientation where minimal elements are represented by white disks and maximal elements by white squares (right).}
    \label{fig:Poset}
\end{figure}

We call a subgraph $H$ of a graph $G$ \emph{dominating} if any vertex of $V(G) \setminus V(H)$ has at least one neighbour in $V(H)$.
Moreover, if $H$ is connected, then it is called a \emph{connected dominating subgraph} of $G$. To prove the theorem of this section, we will use the following lemma, which is a special case of~\cite[Proposition~2]{H-free}.

\begin{lemma}[\cite{H-free}]\label{lem:dom-subgraph}
Let $G$ be a graph with an optimal colouring $c$ of $G$ such that there exists a connected
dominating subgraph $H$ and a connected ordering $\sigma_H$ of $V(H)$ such that the 
greedy colouring of $H$ relative to $\sigma_H$
agrees 
 with $c$ on $V(H)$.
Then, $G$ is not ugly.
\end{lemma}

  An optimal greedy colouring algorithm for colouring comparability graphs is known, see~\cite[Chapter~5.7]{BOOKgolumbic}. This yields an ordering $\sigma$ of the vertices with $\chi(G,\sigma)=\chi(G)$; however $\sigma$ may not be connected. Here we present a connected variant.

\begin{theorem}\label{thm:comp}
No comparability graph is ugly, and a good connected ordering of any connected comparability graph on $n$ vertices and $m$ edges can be computed in time $O(mn)$.
\end{theorem}
\begin{proof}
  It is known~\cite[Chapter~5.7]{BOOKgolumbic} that given a comparability graph $G$ on $n$ vertices and $m$ edges, we can compute in time $O(mn)$
  a partial order $P$ on the vertices of $G$ whose transitive closure yields $G$, as well as a height function $h$ on $P$ defined by $h(v) = 1$ if $v$ is minimal in $P$, and $h(v) = 1 + \max\{h(w)\colon w\prec v\}$ otherwise. 
  This can be done by computing a transitive orientation of $G$ and conducting a Depth-First Search.
  Furthermore, such a function yields an optimal proper colouring of~$G$.
  In the following, by \emph{height} of $P$ we mean $\max_{v\in{V(G)}} h(v)$.

  Let $G$ be a connected comparability graph and $P$ be one of its associated partial orders.
  Clearly the statement holds if the height of $P$ equals $2$, as $G$ is bipartite in that case.
  Let us assume that $P$ is of height $k\geq 3$ and consider the poset $P'$ obtained from $P$ by removing all maximal elements of $P$, as well as the graph $G'$ associated to $P'$ (note that it may not be connected). 
  Since $h$ restricted to $P'$ defines a height function of $P'$, it yields a $(k-1)$-colouring of $G'$.
  We extend it to a colouring of $P$ as follows. First, we colour all maximal elements of $P$ with colour $k$, and then, we swap the colour classes~$2$ and~$k$. 
  Thus, we have obtained an optimal colouring of $G$ where all maximal elements of $P$ are coloured~$2$ and all minimal elements are coloured~$1$. This colouring process is depicted in
  Figure \ref{fig:Posetcolouring}.
  
  Now, observe that the subgraph $H$ induced by the colour classes~$1$ and~$2$ is bipartite and forms a dominating subgraph of $G$, since every element of $P$ that is neither maximal nor minimal is comparable with some maximal and some minimal element.
  
  We furthermore show that $H$ is connected.
  Let us assume that this is not the case and let $A,B$ be two distinct connected components of $H$.
  Since $G$ is connected, there exists some path in $G$ connecting a vertex in $A$ to a vertex in $B$.
  Among all such possible paths between $A$ and $B$, let $Q$ be one with smallest length and call $a\in A$ and $b\in B$ its extremities.
  Clearly we are done if $Q$ is an edge.
  Otherwise, let $x$ be the neighbour of $a$ in $Q$ and $b'$ be its successor (with possibly $b=b'$).
  Two symmetric cases arise depending on whether $a\prec x$ or $x\prec a$.
  Let us assume without loss of generality that $a\prec x$, hence that $a$ is minimal in $P$.
  Since $Q$ is a shortest path, $a$ and $b'$ are non-adjacent, thus incomparable and hence $a\prec x$ and $b'\prec x$.
  Let $x'$ be a maximal element of $P$ such that $x\prec x'$.
  Since $x'\succ a$ we have that $x'$ belongs to $A$ and as $x'\succ b'$ it is connected to $b'$. But then taking $x'b\in E(G)$ and then following the rest of $Q$ from $b'$ to $b$ is a shorter path to reach $b$ from $A$, compared to $Q$ that starts with $ax$ then $xb'$. 
  We have exhibited a path shorter than $Q$ connecting $A$ to $B$, a contradiction to the choice of $Q$.
  
  Now, since $H$ is connected and bipartite, for any connected vertex-ordering of $H$ that starts with a minimal element of $P$, the greedy algorithm produces a colouring that agrees with $c$ on $H$. Hence, we can apply Lemma~\ref{lem:dom-subgraph} to $G$, $c$ and $H$, which shows that $G$ is indeed not ugly.
  As a connected vertex-ordering of $H$ can be obtained in linear time using a standard graph traversal algorithm, and a colouring of $G'$ may be computed in $O(mn)$ time~\cite[Chapter~5.7]{BOOKgolumbic}, we conclude to the desired time bound of $O(mn)$ for the computation of a good connected ordering of $G$.
\end{proof}

\begin{figure}
    \centering
    \includegraphics[page=3,scale=1]{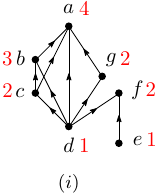}
    \hspace{1cm}
    \includegraphics[page=4,scale=1]{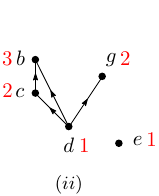}
    \hspace{1cm}
    \includegraphics[page=5,scale=1]{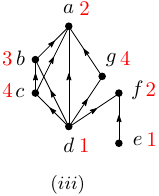}
    \caption{The situation of Theorem~\ref{thm:comp}; $(i)$ a colouring obtained with the
    method of Golumbic $(ii)$ the same colouring restricted to non-maximal elements $(iii)$ the colouring obtained by swapping colours $2$ and $k$.}
    \label{fig:Posetcolouring}
\end{figure}

\subsection{Perfect graphs}\label{sec:perfect}

We now prove our main result, that there are no ugly perfect graphs. This generalizes the same fact which was previously proved for Meyniel graphs~\cite{meyniel} (a class which contains chordal graphs, HHD-free graphs, Gallai graphs, parity graphs, distance-hereditary graphs...) and line graphs of bipartite graphs~\cite{bonamy_21}. Our proof is a generalization of the proof of the latter result by Bonamy, Groenland, Muller, Narboni, Pek\'arek and
Wesolek~\cite[Theorem 2]{bonamy_21}, and our presentation is based on theirs.

\begin{theorem}\label{thm:perfectnotugly}
  No perfect graph is ugly.
\end{theorem}

\begin{proof}
  We will consider only connected perfect graphs, and show that there exists a good connected ordering of their vertices. The proof will use induction on the chromatic number of the
  graphs. As usual for inductive proofs, we shall adapt the
  induction hypothesis: we want it as weak as possible to ease its
  proof, and at the same time as strong as possible since it is
  our basic hypothesis. We shall prove the following
  statement, which implies the theorem.

  \begin{quotation}
    \noindent For any positive integer $k$, any connected $k$-chromatic perfect
    graph $G$ and any vertex $v$ of $G$, there is a
    connected ordering starting with $v$ producing a greedy colouring with
    $k$ colours.
  \end{quotation}

  If $k=1$, then the graph is just a single vertex and the
  statement is true. We now suppose that $k\geq2$ and
  that the induction hypothesis is true for all $k'$ strictly smaller than
  $k$.

  Let $G$ be a connected $k$-chromatic perfect graph with some initial vertex $v$, and $\varphi: V(G) \rightarrow \{1,\ldots,k\}$ be a proper
  $k$-colouring of the vertices of $G$ such that $v$ does not get
  colour $k$ (this is possible since $k\geq2$). For any vertex
  $u$, we say that $v$ {\em reaches} $u$ if there is a path
  $v=s_0,\ldots,s_p=u$ such that for every $1\leq i\leq p$,
  if $\varphi(s_i) = k$, then the edge $s_{i-1}s_i$ is part of a
  $k$-clique. We first prove the following. 

\medskip

  \begin{claim}\label{claim:perfect}
  $G$ has a $k$-colouring such that $v$ reaches all other vertices of $G$.
  \end{claim}
  \claimproof
  Consider a colouring $\varphi$ maximizing the number
  of vertices reached from $v$. Let $A$ be the set of vertices
  reached from $v$ (including $v$) and $B$ the remaining
  vertices. If $B$ is empty, we are done. If not, we build a
  better colouring. In this case, observe that any edge $xy$ from
  $A$ to $B$ must be such that $\varphi(y) = k$ and the edge $xy$
  does not belong to any $k$-clique.

  Let $u$ be some vertex in $B$ 
  such that $u$ has a neighbour in $A$. 
  The graph $G[B]$ induced by $B$ is perfect (since $G$ is
  perfect). Pick any optimal colouring $\rho$ of $G[B]$ such that $u$ does not
  receive colour $k$ and let $S_B$ be the independent set of
  vertices of $B$ getting colour $k$ by $\rho$. Let $S_A$ be the independent set of
  vertices in $A$ getting colour $k$ by $\varphi$. 
  Note that there is no edge between $S_A$ and $B$, so $S_A\cup S_B$ is an independent set: indeed recall that every edge $xy$ with $x\in A$ and $y\in B$ is such that $y$ is coloured $k$ by $\varphi$, but the vertices in $S_A$ are also coloured $k$ by $\varphi$. Thus, $x$ being in $S_A$ would contradict the fact that $\varphi$ is a proper colouring of $G$.

  Since no edge
  between $A$ and $B$ is part of a $k$-clique, each $k$-clique of $G$ is included either in $A$ or in $B$ and thus intersects the set $S_A\cup S_B$. 
  Hence, $G - (S_A \cup S_B)$ has clique number at most $k-1$ and by the perfectness of
  $G$, there is a $(k-1)$-colouring $\gamma$ of $G - (S_A \cup S_B)$. Since $S_A\cup S_B$ is an independent set, we can
  extend $\gamma$ to the whole graph by assigning colour
  $k$ to all vertices in $S_A \cup S_B$. We have that:
  \begin{itemize} 
  \item all vertices in $A$ remain reachable in $\gamma$, as we can consider the same path as for $\varphi$ in $A$, since colour class $k$ is the same in $\varphi$ and $\gamma$;
  \item the vertex $u$ is now reachable, as it has a neighbour in $A$ and is not coloured $k$ by $\gamma$.
  \end{itemize}
  Thus, we have strictly increased the number of reachable vertices,
  which contradicts the choice of $\varphi$. Therefore, there exists $\varphi$ such
  that $v$ reaches the whole graph, and the proof of the claim is complete. \cqedsymbol

  \bigskip

  Let $\varphi$ be a $k$-colouring of $G$ such that $v$ reaches all other vertices of $G$ (obtained from Claim~\ref{claim:perfect}) and let $S$ be the set of vertices coloured $k$ by $\varphi$. The
  graph $G-S$ can be decomposed into connected components
  $C_1,\ldots, C_{\ell}$. Let $C_1$ be the component containing
  $v$. By connectivity of $G$, and after a possible renumbering of
  $C_2, \ldots, C_\ell$, 
  we may find for each index $i$ between $1$ and
  $\ell-1$
  \begin{equation*}
    \text{two vertices } u_i \text{ in } C_1\cup \ldots \cup C_i \text{ and } s_i \text{ in } S \cap N(C_{i+1})
  \end{equation*}
  such that $v$ reaches $s_i$ through $u_i$ (thus, the edge
  $u_is_i$ is part of a $k$-clique).

  Now we can use the induction hypothesis to greedily colour the whole graph $G$ in a connected fashion. Since $C_1$ induces a
  perfect connected graph of chromatic number at most $k-1$,
  by induction, there is a good connected ordering of $C_1$ starting from $v$. This means that
  $u_1$ is coloured. Since $u_1s_1$ is in a $k$-clique, the
  other members of this clique (except $s_1$) are in $C_1$. Thus,
  they use all colours among ${1,\ldots,k-1}$. The greedy colouring
  continuing with $s_1$ will then assign colour $k$ to it. Now, $s_1$ has
  a neighbour in $C_2$. By induction, there is a connected greedy
  $(k-1)$-colouring of $G[C_2]$ starting with colour~1 from any
  vertex, so we can colour $G[C_2]$. We iterate the process through all connected components. At last, we colour the uncoloured vertices of $S$. This process yields a connected greedy $k$-colouring of $G$.
\end{proof}

Note that the proof of Theorem~\ref{thm:perfectnotugly} is constructive; it directly yields an algorithm for finding a good connected ordering of any input connected perfect graph $G$ with $n$ vertices.
But in order to be able to find a good connected ordering of $G$, we must be able to compute a $k$-colouring of $G$ such that $v$ reaches all other vertices of $G$ (Claim~\ref{claim:perfect} in the proof of Theorem~\ref{thm:perfectnotugly}).
This can be done by computing an optimal colouring of $G$ (since $G$ is perfect, a colouring of $G$ using $\omega(G)$ colours can be found in polynomial time $O(n^c)$ for some $c\in\mathbb{N}$ using the ellipsoid method~\cite{grotschel}), and repeatedly applying the argument of the proof of Claim~\ref{claim:perfect} to extend the set of vertices that can be reached from $v$.
This is formalized by Algorithm~\ref{alg:perfect-reached}. 
The size of the maximum clique of Line~\ref{line:omega-computation} is computed in $O(n^c)$ time using the algorithm in~\cite{grotschel}, bringing the total time complexity of Algorithm~\ref{alg:perfect-reached} to $O(n^{c+2})$.

    \begin{algorithm} \label{alg:perfect-reached}
    \DontPrintSemicolon%
    \SetNoFillComment%
    \caption{Given a perfect graph $G$ on $n$ vertices, an optimal colouring of $G$ and a vertex~$v$, compute
    the set of vertices that are reached from $v$. } 
    \KwIn{A perfect graph $G = (V, E)$, an optimal $k$-colouring of the vertices of $G$
    where $k = \omega(G)$ and a vertex $v$ of $G$.}
    \KwOut{The set of vertices reachable from $v$.}
    \label{algo3}
    
    \tcc{Construct the directed graph $D$ of direct reachability from $v$.}
     
     Let $D$ be a directed graph with $V(D)=V(G)$ and no arc.
     
     \For{each edge $uw\in E(G)$}
     {%
        \lIf{neither $u$ nor $w$ is coloured $k$}{Add both $u \rightarrow w$ and $w \rightarrow u$.}
        \Else
        {%
            Without loss of generality, let $u$ be the vertex coloured $k$.
            
            Add the arc $u \rightarrow w$ to $D$.
            
            Let $q = \omega(G[N(u) \cap N(w)])$.\label{line:omega-computation}
            
            \lIf{$q = k-2$}{Add the arc $w \rightarrow u$ to $D$.}
            
        }

     }
     
    Let $A$ be the set of vertices visited during a traversal of $D$ starting from $v$.
     
    \Return{$A$}
    
    \end{algorithm}

Then, Algorithm~\ref{alg:perfect-reached} is used as a sub-routine in Algorithm~\ref{alg:perfect-good-ordering} to compute an optimal connected colouring of $G$, leading to the next corollary.

    \begin{algorithm} \label{alg:perfect-good-ordering}
    \DontPrintSemicolon%
    \SetNoFillComment%
    \caption{Compute an optimal connected colouring of a given perfect graph.}
    \KwIn{A perfect graph $G = (V, E)$.}
    \KwOut{A good connected ordering of the vertices of $G$.}
    
    Let $\varphi : V \rightarrow \mathbb{N}$ be an optimal $k$-colouring of $G$ where $k = \omega(G)$.
    
    \tcc{Computed with a complexity of $|V|^{c}$ for some constant $c\in \mathbb{N}$.}
    
    Let $v$ be a vertex of $V$ such that $\varphi(v)\neq k$ and $A$ be the set of vertices reached from $v$.
    
    \tcc{Computed with Algorithm~\ref{algo3} as a subroutine.}
    
    \While{there exists some vertex not in $A$}
    {%
        Let $u$ be a vertex not in $A$ with a neighbour in $A$ and $B = V \setminus A$.
        
        Compute an optimal colouring $\rho$ of $G[B]$ with at most $k$ colours.
        
        Let $S_A$ (resp. $S_B$) be the independent set of vertices in $A$ (resp. $B$) coloured $k$ by $\varphi$ (resp. coloured $k$ by $\rho$).
        
        Let $G' = G[V \setminus (S_A \cup S_B)]$ and compute an optimal colouring $\gamma$ of $G'$ with
        at most $k-1$ colours.
        
        Extend $\gamma$ by assigning colour $k$ to the vertices in $S_A \cup S_B$.
        
        $\varphi \leftarrow \gamma$
        
        Recompute the set $A$ of vertices reached by $v$ in $\varphi$.
    }
    
    \tcc{From now on, $v$ reaches all the vertices of $G$.}
    Let $S$ be the set of vertices coloured $k$ by $\varphi$ in $G$.
    
    Let $C_1, \dots, C_{\ell}$ be the connected components of $G[V \setminus S]$.
    
    \tcc{Computed in $O(|E|)$.}
    
    
    $v_1 \leftarrow v$
    
    \For{$i$ in  $\{1, \dots, \ell \}$}
    {%
         \tcc{Recursive call :}
         Compute a good connected ordering of $G[C_i]$ starting in $v_i$, using $k-1$ colours; add it to the final connected ordering $\sigma$ to be outputted.
        
        \If{$i\neq \ell$}{ 
        Let $u_i$ in $C_1 \cup \ldots \cup C_i$ and $s_i$ in $S \cap N(C_{i+1})$ be two
        vertices such that $v$ reaches $s_i$ through $u_i$.
        
         Assign colour $k$ to $s_i$ and add it to $\sigma$.
         
         Let $v_{i+1}$ be a neighbour of $s_i$ in $C_{i+1}$.
        }
         
    }

    Add the uncoloured vertices of $S$ to the ordering $\sigma$
    
    \Return{the good connected ordering $\sigma$ of $G$} using $k$ colours.
    
    \end{algorithm}
    
\begin{corollary}
    A good connected ordering of any connected perfect graph on $n$ vertices can be computed in time $O(n^{c+4})$ provided that an optimal colouring of a perfect graph can be obtained in $O(n^c)$ time.
\end{corollary}
    
Note that the time bound of the above corollary relies to date on the complexity of the polynomial-time algorithm from~\cite{grotschel}, whose precise exponent has not been made explicit by the authors and which is most probably large. 
This is in contrast to the algorithm for comparability graphs given by Theorem~\ref{thm:comp} which runs in $O(mn)$ time.
Concerning other subclasses of interest, as mentioned in the introduction, the same task can be done in time $O(m+n)$ for chordal graphs using the LexBFS algorithm~\cite{chordal}, for Meyniel graphs this can be done in time $O(n^2)$ using a variant of LexBFS~\cite{meyniel}, and a careful inspection of the proof in~\cite{bonamy_21} gives an $O(n^4)$ algorithm for line graphs of bipartite graphs.

\paragraph{Acknowledgements.} This research was partially financed by the French government IDEX-ISITE initiative 16-IDEX-0001 (CAP 20-25) and by the ANR project GRALMECO (ANR-21-CE48-0004). We are thankful to all participants of the 2018 AlCoLoCo problem seminars and the 2018 Recolles workshop, where this research was started. In particular, we thank Giacomo Kahn and Armen Petrossian for preliminary discussions. We also thank the referees for their comments that helped improve the paper.

\end{document}